\documentclass{article}
\PassOptionsToPackage{compatV3}{fancyhdr}
\usepackage{iclr2027_conference,times}

\usepackage[utf8]{inputenc}
\usepackage[T1]{fontenc}
\usepackage{xcolor}
\usepackage[colorlinks=true, linkcolor={blue!60!black}, citecolor={blue!60!black}, urlcolor={blue!60!black}]{hyperref}
\usepackage{url}
\usepackage{booktabs}
\usepackage{multirow}
\usepackage{amsfonts}
\usepackage{nicefrac}
\usepackage{microtype}

\usepackage{graphicx}
\usepackage{wrapfig}
\usepackage{algorithm}
\usepackage{algorithmic}
\usepackage{enumitem}
\usepackage{subcaption}
\usepackage{amsmath}
\usepackage{amssymb}
\usepackage{mathtools}
\usepackage{amsthm}
\usepackage{bm}
\usepackage{xspace}

\newcommand{\AutoAdjust}[3]{{ \mathchoice{ \left #1 #2  \right #3}{#1 #2 #3}{#1 #2 #3}{#1 #2 #3} }}

\newcommand{\InNorms}[1]{\AutoAdjust{\|}{#1}{\|}}

\newcommand{\F}{Figure}
\newcommand{\T}{Table}
\newcommand{\A}{Algorithm}

\theoremstyle{plain}
\newtheorem{theorem}{Theorem}[section]
\newtheorem{proposition}[theorem]{Proposition}

\theoremstyle{definition}
\newtheorem{definition}[theorem]{Definition}

\theoremstyle{remark}

\newcommand\vecx{{\ensuremath{\bm x}\xspace}}

\newcommand{\RETURN}{\STATE \textbf{Return }}

\title{Guided Diffusion for Trustworthy Synthetic Text Generation}

\author{Qi Pang$^{1,2}$\thanks{Part of this work was done during an internship at NVIDIA. Email: \texttt{qipang@cmu.edu}.}\hspace{0.45em}, Hanshen Xiao$^{1,3}$, Andre Manoel$^{1}$, Lipika Ramaswamy$^{1}$, \\
  \textbf{Wenting Zheng$^{2}$, Virginia Smith$^{2}$, G.~Edward Suh$^{1}$} \\
  $^{1}$NVIDIA \quad $^{2}$Carnegie Mellon University \quad $^{3}$Purdue University \\
}

\iclrfinalcopy
\begin{document}

\maketitle
\lhead{}
\renewcommand{\headrulewidth}{0pt}

\begin{abstract}
The rapid evolution of Large Language Models is fueled by massive datasets, yet this reliance poses severe privacy risks, as models can inadvertently memorize and even directly output sensitive training data.
Synthetic text generation offers a promising solution to this challenge; however, existing methods struggle to reconcile rigorous privacy guarantees with high data utility.
In this paper, we propose a new paradigm for private synthetic text generation: Guided Diffusion.
Leveraging the unique steerability of Diffusion Language Models, we introduce novel sampling and re-masking procedures that actively guide the iterative denoising steps.
These mechanisms steer the generative process to reconstruct fluent text that adheres to a target semantic distribution.
We demonstrate that our framework is naturally compatible with Differential Privacy (DP): once the semantic embeddings of the private corpus are protected, the DP guarantee carries over to the generated text by post-processing.
Empirically, we show that our approach generates high-quality synthetic text with precise semantic alignment.
Furthermore, we validate its effectiveness in DP synthetic text generation: across two datasets and two diffusion backbones, it outperforms Private Evolution and DP-SGD-based fine-tuning in distributional fidelity and downstream utility, while generating significantly faster than Private Evolution.
We further extend the framework to Data-Specific Indistinguishability (DSI), which adds a formal certificate against data poisoning and memorization at a small utility cost.
\end{abstract}

\section{Introduction}
\label{sec:intro}

The rapid evolution of Large Language Models (LLMs) has fundamentally reshaped Natural Language Processing, a paradigm shift driven primarily by the massive scale of training data.
However, this reliance on web-scale data accumulation comes with a steep risk of exposing sensitive information.
Real-world datasets, ranging from private medical records and financial reports to personal communication logs, are invaluable for training robust models, but they often contain sensitive personal information that must not be exposed.
Recent studies have demonstrated that generative models are prone to memorization~\citep{carlini2021extracting}, where they can inadvertently regenerate training samples verbatim.
This vulnerability poses severe ethical concerns and violates strict privacy regulations such as GDPR~\citep{voigt2017eu} and CCPA~\citep{ccpa2018}.

To address this, synthetic text generation has emerged as a promising solution.
Ideally, a synthetic dataset mirrors the statistical properties and semantic richness of the real distribution while ensuring that no individual data point is leaked.
However, existing approaches typically struggle to balance these objectives, falling into two broad categories.
The first relies on \textit{heuristic sanitization}, such as entity masking or scrubbing~\citep{lison-etal-2021-anonymisation}.
While these methods preserve linguistic fluency, they lack formal guarantees and have been repeatedly shown to fail against sophisticated reconstruction attacks~\citep{carlini2021extracting, balle2022reconstructing, zhou2024model}.
The second category employs \textit{rigorous formalisms} like Differential Privacy (DP), often applied via DP-SGD~\citep{abadi2016deep}.
While provably privacy-preserving, these methods frequently result in a catastrophic loss of utility when applied to text, yielding outputs that are repetitive or incoherent.

One fundamental limitation behind such heavy utility compromise is that provable privacy guarantees such as DP rely on noise to ensure {\em statistical closeness}.
For example, DP requires statistical indistinguishability between the leakage produced by any two adjacent datasets for individual privacy.
To support this, existing noise mechanisms require calibration to the absolute divergence or distance.
However, text representations in vocabulary space have extremely high variance.
For instance, the sentences ``The movie is excellent'' and ``This is a superb film'' are semantically identical yet possess disjoint vocabulary representations (orthogonal one-hot vectors).
Standard distance metrics in this space fail to capture semantic proximity, and the required noise scale to mask these differences is often prohibitively large, destroying the linguistic structure of the output.

In this paper, we propose a novel framework for synthetic text generation that seamlessly integrates with off-the-shelf DP clustering algorithms to provide rigorous privacy guarantees.
Our core insight is that, while raw text resides in a high-dimensional, discrete space, its semantic embedding can be captured on a low-dimensional, continuous manifold.
We observe that DP mechanisms can achieve significantly higher utility if applied within this compact embedding space instead of the sparse vocabulary space.
The challenge then shifts to a generation problem:
\begin{center}
    \emph{How to convert privatized, low-dimensional semantic embeddings into coherent, discrete text?}
\end{center}

While latent-conditioned generation is well-established in computer vision (e.g., image editing~\citep{rombach2022high, ye2023ip}), transferring this success to text is challenging due to the discrete nature of language.
Auto-regressive (AR) LLMs, which generate tokens sequentially, are notoriously difficult to steer globally using a fixed embedding vector.
To bridge this gap, we leverage the unique properties of Diffusion Language Models~\citep[DLMs;][]{li2022diffusion, sahoo2024simple, nie2025large, you2025llada, song2025seed, yang2025mmada}.
Unlike AR models, DLMs predict blocks of tokens (sentences) through an iterative denoising process.
This paradigm allows us to effectively steer the generation by minimizing the distance between the generated sequence and the target semantic embedding at each diffusion step.

Specifically, we introduce a Guided Diffusion framework that generates coherent text conditioned on semantic embeddings.
To enable precise control, we propose two novel components for masked DLMs: {\em semantic-aware sampling} and {\em semantic-aware re-masking}.
They allow the DLM to align its discrete outputs with the continuous embedding guidance signal.
By integrating DP mechanisms directly into the embedding stage, we present the first framework capable of utilizing DLMs for synthetic data generation with rigorous privacy protection.
Beyond DP, the same framework accommodates Data-Specific Indistinguishability (DSI)~\citep{xiao2025trustworthy}, which targets risks such as data poisoning and unintended memorization (\S~\ref{sec:dsisyn}).
In summary, we make the following contributions:
\begin{itemize}[leftmargin=8pt, topsep=2pt]
    \item \textbf{Guided Diffusion Framework}: We propose a novel framework that, for the first time, exploits the steering capabilities of DLMs to generate coherent text conditioned on semantic embeddings.
    \item \textbf{DP and DSI Integration}: Our Guided Diffusion framework naturally accommodates off-the-shelf DP clustering mechanisms in the embedding space, and the DP guarantee carries over to the synthetic text by post-processing. The same framework also supports DSI aggregation, which provides formal guarantees against corrupted data sources and memorization (\S~\ref{sec:dsisyn}).
    \item \textbf{Empirical Effectiveness}: We validate our framework through extensive experiments, demonstrating superior synthetic data quality and strong downstream utility in DP synthetic text generation: across two datasets and two diffusion backbones, it outperforms Private Evolution and DP-SGD-based baselines in distributional fidelity and downstream utility, while generating substantially faster than Private Evolution. We will release our code publicly.
\end{itemize}

\section{Preliminaries}
\label{sec:preliminaries}

\noindent\textbf{Diffusion Language Models.}~Diffusion Language Models (DLMs) are non-autoregressive generative models that synthesize text through an iterative denoising process.
Unlike AR LLMs that generate tokens sequentially from left to right, DLMs generate a sequence of $L$ tokens in parallel by progressively refining a noisy input state.

We focus on the denoising process, illustrated in \F~\ref{fig:dlm-pipeline}.
Let $\vecx_t$ denote the token sequence at step $t$, where $t$ goes from $T$ down to $0$.
The process begins with $\vecx_T$, a fully masked sequence consisting entirely of the special token $\texttt{[MASK]}$.
The core component is a mask predictor model $\mathcal{M}_{\theta}$, which takes a partially masked sequence $\vecx_t$ as input and outputs a probability matrix $\mathbf{P}_t \in \mathbb{R}^{L \times |\mathcal{V}|}$ for $L$ positions over the vocabulary $\mathcal{V}$.
The transition from step $t$ to $t-1$ involves two critical operations:

\noindent\textit{Sampling.}~Given the current state $\vecx_t$ (e.g., ``{This paper $\texttt{[MASK]}$ a $\texttt{[MASK]}$ approach...}''), the model $\mathcal{M}_{\theta}$ predicts the probability distribution for every token position.
From this probability matrix $\mathbf{P}_t$, we independently sample a candidate token for each masked position to form a fully populated candidate sequence $\vecx'_{t-1}$, as an intermediate state (see \F~\ref{fig:dlm-pipeline}).

\noindent\textit{Re-Masking.}~Since early predictions may be inaccurate, we do not accept all sampled tokens immediately. We obtain a confidence score (probability) for each token in $\vecx'_{t-1}$.
As shown in \F~\ref{fig:dlm-pipeline}, tokens with high confidence (depicted in green) are retained, while tokens with low confidence (depicted in red) are reverted to the $\texttt{[MASK]}$ token.
This produces the input state $\vecx_{t-1}$ for the next iteration. This cycle repeats until $t=0$ and all masks are removed, resulting in the final coherent sentence.

\begin{figure}[!t]
\centering
\includegraphics[width=1.0\linewidth]{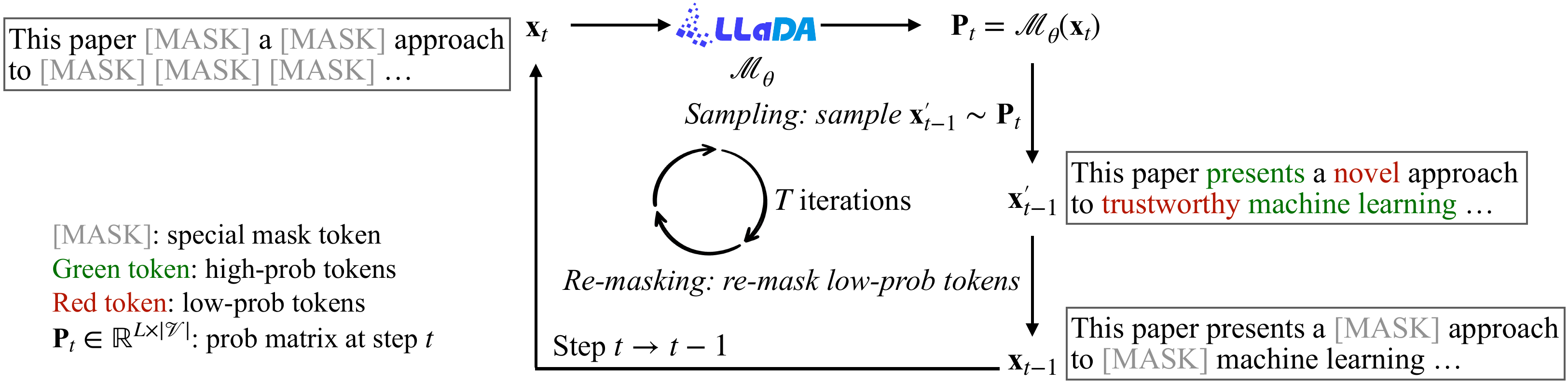}
\caption{High-level denoising pipeline for DLMs. At each step $t$, the model predicts token probabilities for the masked input $\vecx_t$. A candidate sequence $\vecx'_{t-1}$ is sampled, and low-confidence tokens are re-masked to produce the refined state $\vecx_{t-1}$.}
\label{fig:dlm-pipeline}
\end{figure}

\noindent\textbf{Differential Privacy.}~Differential Privacy~\citep[DP;][]{dwork2006calibrating} bounds the influence of any single record on a randomized algorithm's output. Throughout, two datasets $\mathcal{D}, \mathcal{D}'$ are \emph{neighboring} if one is obtained from the other by adding or removing a single record (add/remove adjacency). We present the formal definition for $(\varepsilon,\delta)$-DP.
\begin{definition}[$(\varepsilon,\delta)$-DP]
A randomized algorithm $\mathcal{A}$ satisfies $(\varepsilon, \delta)$-DP if for all neighboring datasets $\mathcal{D}, \mathcal{D}'$ and all measurable sets $S$,
\begin{equation*}
\Pr[\mathcal{A}(\mathcal{D}) \in S] \leq e^{\varepsilon} \Pr[\mathcal{A}(\mathcal{D}') \in S] + \delta.
\end{equation*}
\end{definition}
Smaller $\varepsilon$ and $\delta$ correspond to stronger privacy guarantees: an adversary distinguishing two neighboring datasets must accept a higher Type-I (or Type-II) error rate at any fixed Type-II (or Type-I) budget.
Our framework relies on three standard properties.

(i) The \emph{Gaussian mechanism}~\citep{balle2018improving} $\mathcal{A}(\mathcal{D}) = f(\mathcal{D}) + \mathcal{N}(\bm{0}, \sigma^2 I_d)$ on a vector-valued query $f$ with outputs in $\mathbb{R}^d$ and $L_2$ sensitivity $\Delta$ achieves $(\varepsilon, \delta)$-DP with
\begin{equation*}
\delta = \Phi\!\left(\tfrac{\Delta}{2\sigma} - \tfrac{\varepsilon \sigma}{\Delta}\right) - e^{\varepsilon}\, \Phi\!\left(-\tfrac{\Delta}{2\sigma} - \tfrac{\varepsilon \sigma}{\Delta}\right),
\end{equation*}
where $\Phi(\cdot)$ is the standard normal cumulative distribution function.

(ii) \emph{Post-processing invariance}: any function applied to the output of an $(\varepsilon, \delta)$-DP algorithm that does not access the private input incurs no additional privacy loss. In other words, once the noisy embeddings are released under DP, any text generated from them alone, without touching the private data again, inherits the same DP guarantee.

(iii) \emph{Parallel composition}: if a dataset is partitioned into disjoint subsets by a public rule and an $(\varepsilon, \delta)$-DP algorithm is run on each subset, the joint release is $(\varepsilon, \delta)$-DP under add/remove adjacency, since a single record affects only one subset. We use this to run DP clustering separately within each label class. Under replace-one adjacency, where a record may change its label and thereby move between two subsets, the same joint release is $(2\varepsilon, 2\delta)$-DP.

\section{Guided Diffusion Language Models}
\label{sec:guided-diffusion-design}

In this section, we present our framework of Guided Diffusion.
The primary objective is to steer the generative process such that the resulting text $\vecx$ achieves high similarity to a target semantic embedding $E_{\text{tgt}}$, while preserving linguistic fluency.
To achieve this, we introduce two novel mechanisms -- \textit{Semantic-Aware Sampling} and \textit{Semantic-Aware Re-Masking} -- and integrate them directly into the iterative denoising pipeline of standard DLMs.

Let $\mathcal{M}_\theta$ denote the diffusion mask predictor.
At each diffusion step $t$, the model receives a partially masked sequence $\vecx_t$ and outputs a probability matrix $\mathbf{P}_t \in \mathbb{R}^{L \times |\mathcal{V}|}$, where each row represents the probability distribution over the vocabulary $\mathcal{V}$ for the corresponding token position.
To output a coherent token block of length $L$, DLMs typically employ sampling and re-masking steps and iteratively remove the noise, which is the $\texttt{[MASK]}$ token, in the input.
Now we introduce our Semantic-Aware Sampling and Semantic-Aware Re-Masking designs.

\subsection{Semantic-Aware Sampling}
\label{subsec:guided-sampling}
Standard DLMs employ independent token sampling, where a candidate sequence is formed by sampling from the probability distribution at each token position independently.
While efficient, this approach is agnostic to the global semantic trajectory required for our specific guidance task.

To enforce semantic alignment, we propose a two-stage reranking procedure that selects among sampled candidates by balancing exploration, fluency, and semantic distance.
We present the algorithm in \A~\ref{alg:sampling}.
Specifically, at diffusion step $t$, given the probability matrix $\mathbf{P}_t$, we first generate a pool of $\alpha N$ candidate sequences $\{\hat{\vecx}'^{(k)}\}_{k=1}^{\alpha N}$ via independent sampling: $\hat{\vecx}'^{(k)} \sim \mathbf{P}_t$.

We enforce that all $\alpha N$ candidate sequences differ by at least one token.
We then evaluate these candidates using a hybrid scoring function to filter for quality and diversity:
\begin{equation*}
    S_k = a \cdot \log p_\theta(\hat{\vecx}'^{(k)}) + b \cdot \mathrm{Ham}(\hat{\vecx}'^{(k)}), \quad k \in [\alpha N],
\end{equation*}
where $\log p_\theta(\cdot)$ denotes the log-probability of the sequence under the model $\mathcal{M}_\theta$, and $\mathrm{Ham}(\cdot)$ measures the Hamming distance between $\hat{\vecx}'^{(k)}$ and the highest-likelihood sequence from $\mathbf{P}_t$.
$a$ and $b$ are scalar weights for the two scoring terms; the values used throughout our experiments are reported in Appendix~\ref{sec:implementation}, and Appendix~\ref{sec:appendix-sensitivity} shows that generation quality is insensitive to the ratio $b/a$ over the tested range.
We sort the candidates by $S_k$ and retain the top $N$ sequences, denoted as $\{\hat{\vecx}^{(k)}\}_{k=1}^{N}$.

\begin{algorithm}[!h]
\caption{Semantic-Aware Sampling}
\label{alg:sampling}
\begin{algorithmic}[1]
\REQUIRE Prob. matrix $\mathbf{P}_t$, target embedding $E_{\text{tgt}}$, constants $\alpha, a, b, N$, embedding model $\mathcal{E}$
\STATE Generate pool $\mathcal{C} = \{\hat{\vecx}'^{(k)}\}_{k=1}^{\alpha N}$ via $\hat{\vecx}'^{(k)} \sim \mathbf{P}_t$
\STATE Compute $S_k = a \cdot \log p_\theta(\hat{\vecx}'^{(k)}) + b \cdot \mathrm{Ham}(\hat{\vecx}'^{(k)}), k \in [\alpha N]$
\STATE Filter $\mathcal{C}$ to top $N$ candidates $\{\hat{\vecx}^{(k)}\}$ based on $S_k$
\STATE Compute $E^{(k)} = \mathcal{E}(\hat{\vecx}^{(k)})$ for filtered candidates
\STATE Select $\hat{\vecx}^* = \hat{\vecx}^{(k^*)}$, $k^* = \arg \min_k ||E^{(k)} - E_{\text{tgt}}||_2$
\RETURN $\hat{\vecx}^*$
\end{algorithmic}
\end{algorithm}

Next, we select the candidate that best aligns with the target semantics. For each retained candidate $\hat{\vecx}^{(k)}$, we compute its semantic embedding $E^{(k)} = \mathcal{E}(\hat{\vecx}^{(k)})$, where $\mathcal{E}(\cdot)$ is the pre-trained semantic encoder.
We evaluate the alignment with the target embedding $E_{\text{tgt}}$ using the Euclidean distance $d_k = ||E^{(k)} - E_{\text{tgt}}||_2$.
The optimal candidate $\hat{\vecx}^*$ is selected as the minimizer of this distance:
\begin{equation*}
    \hat{\vecx}^* = \hat{\vecx}^{(k^*)} \quad \text{where} \quad k^* = \arg \min_{k \in [N]} d_k
\end{equation*}
This selected sequence $\hat{\vecx}^*$ serves as the base estimation for the current step, effectively steering the generation manifold toward the target semantics before the re-masking phase.

\subsection{Semantic-Aware Re-Masking}
\label{subsec:guided-remasking}

After obtaining the candidate sequence $\hat{\vecx}^*$, the denoising process requires a re-masking step to transition to the next noisy state $\vecx_{t-1}$.
Standard strategies typically mask tokens based on low model confidence (low prediction probability) or simple random selection.
However, these heuristics may inadvertently discard tokens that are critical for semantic alignment or retain tokens that steer the text away from the target semantic embedding $E_{\text{tgt}}$.

To address this issue, we propose a contribution-based re-masking algorithm.
Our insight is to strictly retain tokens that maximally contribute to minimizing the distance to the target semantic embedding, while masking those that are less relevant.

\begin{algorithm}[!h]
\caption{Semantic-Aware Re-Masking}
\label{alg:remasking}
\begin{algorithmic}[1]
\REQUIRE Sequence $\hat{\vecx}^*$, target embedding $E_{\text{tgt}}$, embedding model $\mathcal{E}$, denoise schedule $\rho_t$
\STATE Compute base embedding $E_{\text{base}} = \mathcal{E}(\hat{\vecx}^*)$
\STATE \textbf{for} position $i = 1$ \textbf{to} $L$ \textbf{do}
\STATE \quad Ablate token $w_i \to \texttt{[MASK]}$ to get $\hat{\vecx}^*_{\setminus i}$
\STATE \quad Compute importance $\Delta_i = ||\mathcal{E}(\hat{\vecx}^*_{\setminus i}) - E_{\text{tgt}}||_2 - ||E_{\text{base}} - E_{\text{tgt}}||_2$
\STATE Identify indices $\mathcal{I}_{\text{mask}}$ corresponding to the $\rho_t/2$ smallest scores $\Delta_i$, where $|\mathcal{I}_{\text{mask}}| = \rho_t/2$
\STATE Update $\vecx_{t-1}$ by setting $w_i \leftarrow \texttt{[MASK]}$ for all $i \in \mathcal{I}_{\text{mask}}$
\STATE Identify indices $\mathcal{I}_{\text{low-prob}}$ corresponding to the $\rho_t / 2$ lowest prediction probabilities, where $|\mathcal{I}_{\text{low-prob}}| = \rho_t/2$ and $\mathcal{I}_{\text{mask}} \cap \mathcal{I}_{\text{low-prob}} = \emptyset$
\STATE Update $\vecx_{t-1}$ by setting $w_i \leftarrow \texttt{[MASK]}$ for all $i \in \mathcal{I}_{\text{low-prob}}$
\RETURN $\vecx_{t-1}$
\end{algorithmic}
\end{algorithm}

We present the algorithm in \A~\ref{alg:remasking}.
Let $w_1, w_2, \cdots, w_L$ be the $L$ tokens in $\hat{\vecx}^*$: $\hat{\vecx}^* = [w_1, w_2, \dots, w_L]$ and let $E_{\text{base}} = \mathcal{E}(\hat{\vecx}^*)$ denote its embedding.
To quantify the contribution of the $i$-th token $w_i$, we perform a granular ablation: we replace $w_i$ with the special $\texttt{[MASK]}$ token to create a perturbed sequence $\hat{\vecx}^*_{\setminus i}$ and compute its embedding $E_{\setminus i} = \mathcal{E}(\hat{\vecx}^*_{\setminus i})$.
We define the importance score $\Delta_i$ as the \emph{change} in distance to the target induced by ablating $w_i$:
\begin{equation*}
    \Delta_i = \InNorms{E_{\setminus i} - E_{\text{tgt}}}_2 - \InNorms{E_{\text{base}} - E_{\text{tgt}}}_2, \quad i \in [L]
\end{equation*}
A large positive $\Delta_i$ means masking $w_i$ moves the sequence away from the target, indicating that $w_i$ is critical to alignment. A near-zero $\Delta_i$ means $w_i$ is irrelevant to alignment, and a negative $\Delta_i$ means masking $w_i$ actually \emph{improves} alignment, so the token is beneficial to mask.
In accordance with the diffusion noise schedule, we retain the tokens with the largest $\Delta_i$ and re-mask the remaining tokens.
Specifically, at step $t$, when the diffusion noise schedule calls for masking $\rho_t$ tokens, we first re-mask the $\rho_t/2$ tokens with the smallest $\Delta_i$ and then re-mask an additional $\rho_t/2$ low-probability tokens.
This simultaneously ensures semantic alignment and effective denoising.

By iterating these semantic-aware sampling and re-masking steps, we ensure that the generative process converges to a coherent sentence where the corresponding semantic embedding is close to the target reference embedding vector $E_{\text{tgt}}$.

\noindent\textbf{Necessity of Both Components.}~Semantic-aware sampling and re-masking play complementary roles: sampling steers candidate selection toward the target manifold, while re-masking preserves semantically critical tokens.
We verify that disabling either component substantially degrades both semantic alignment (\S~\ref{subsec:guided-eval}) and downstream DP utility (\S~\ref{subsec:dp-eval}), confirming both stages are essential.

\subsection{Semantic Alignment and Quality Preservation of Guided Diffusion}
\label{subsec:guided-eval}

To demonstrate that Guided Diffusion can effectively generate high-quality sentences while closely aligning with a target semantic embedding, we conduct a comprehensive evaluation on standard benchmarks.

\noindent\textbf{Experimental Setup.}~We use two datasets covering topical and sentiment-based steering: \texttt{Yahoo Answers Topics}~\citep{zhang2015character} (10 topic categories) and \texttt{IMDB}~\citep{imdbdataset} (rich-sentiment movie reviews); from each we randomly draw $1{,}000$ reference sentences exceeding $256$ tokens as semantic targets. We employ \texttt{LLaDA-8B-Instruct}~\citep{nie2025large} as the backbone DLM, with target embeddings produced by \texttt{sentence-t5-base}~\citep{ni2022sentence} for \texttt{Yahoo} and \texttt{stsb-roberta-base-v2}~\citep{stsbrobert} for \texttt{IMDB}. We report (i) \textit{semantic alignment} as cosine similarity between generated and target embeddings, and (ii) \textit{generation quality} as perplexity (PPL) computed with Qwen3-8B~\citep{yang2025qwen3}. We further include an ablation isolating the contribution of semantic-aware sampling and re-masking.

\begin{wrapfigure}{r}{0.58\linewidth}
    \centering
    \includegraphics[width=\linewidth]{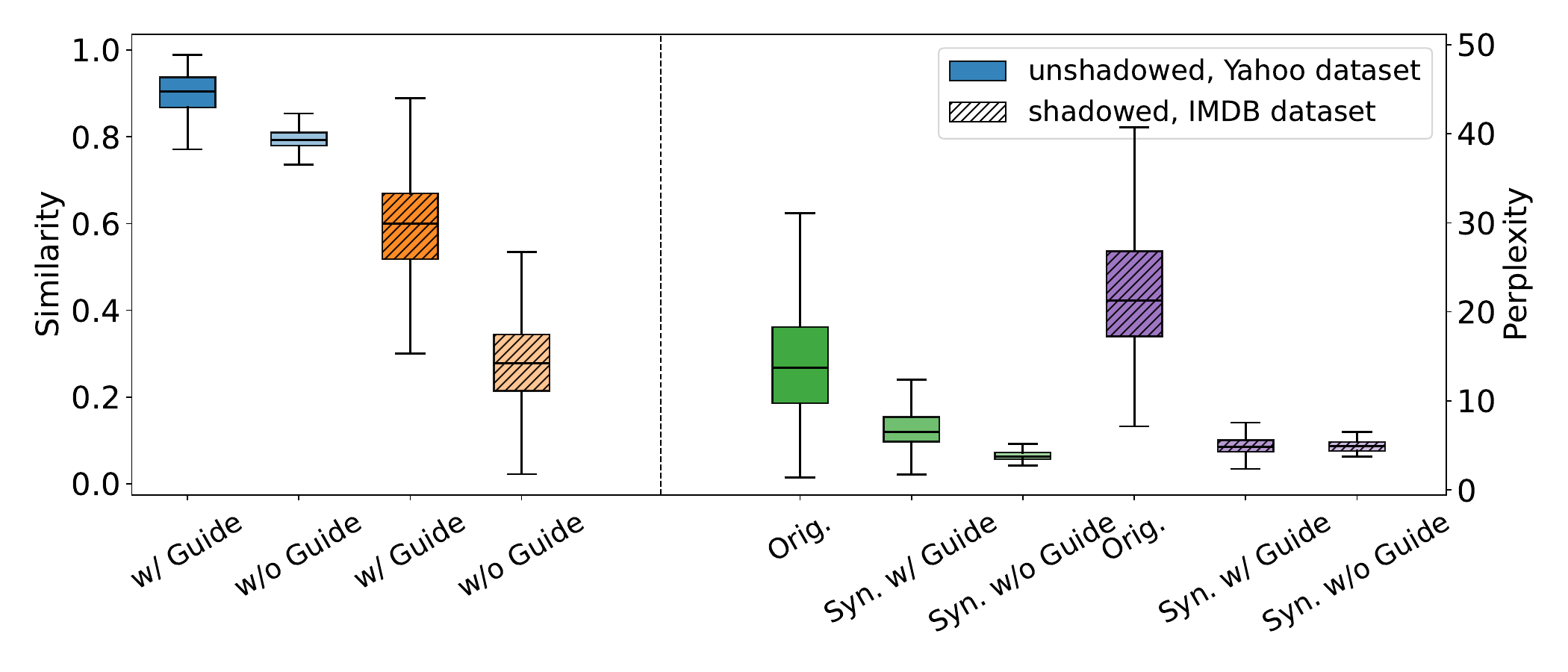}
    \vspace{-18pt}
    \caption{Evaluation of semantic alignment and generation quality. \textbf{Left:} Cosine similarity between generated and target embeddings. Guided Diffusion significantly improves alignment compared to unguided generation. \textbf{Right:} Perplexity (PPL) under Qwen3-8B of the original references (Orig.) and of guided and unguided generations. Guided text remains fluent: on \texttt{IMDB} its PPL matches unguided generation, and on \texttt{Yahoo} it is higher than unguided generation but still well below the references.}
    \label{fig:dis-ppl}
    \vspace{-14pt}
\end{wrapfigure}

\noindent\textbf{Evaluation Results.}~The results are presented in \F~\ref{fig:dis-ppl} and \T~\ref{tab:ablation-semantic}. Across both datasets, Guided Diffusion achieves strong semantic alignment with the target embeddings without sacrificing much on fluency, and both semantic-aware sampling and re-masking are necessary to reach this quality.

\textbf{\textit{Semantic Alignment.}}~\F~\ref{fig:dis-ppl} (Left) shows that our Guided Diffusion achieves high fidelity to the target embeddings.
The median cosine similarity between synthetic and original embeddings is $0.91$ for \texttt{Yahoo} and $0.60$ for \texttt{IMDB}.
In contrast, unguided generation results in significantly lower semantic alignment (e.g., a median of $0.28$ for \texttt{IMDB}), confirming that our Guided Diffusion can effectively preserve the high fidelity to the target embedding.

\textbf{\textit{Generation Quality.}}~In \F~\ref{fig:dis-ppl} (Right), we demonstrate that Guided Diffusion preserves the linguistic fluency.
The synthetic samples remain fluent, with median PPL scores of $6.51$ (\texttt{Yahoo}) and $4.83$ (\texttt{IMDB}).
On \texttt{IMDB} this matches unguided generation ($4.89$); on \texttt{Yahoo} guided text has a higher PPL than unguided text, since steering toward a specific target constrains the sampler, but both remain well below the PPL of the human references. Our sampling and re-masking interventions therefore preserve fluency.

\begin{wraptable}{r}{0.62\linewidth}
\vspace{-12pt}
\caption{Ablation of semantic-aware sampling and re-masking. We report the cosine similarity between generated and target embeddings on \texttt{Yahoo} and \texttt{IMDB} as mean$_{\pm\text{std}}$; bold marks the best result. Disabling either substantially degrades alignment, confirming both stages are essential.}
\label{tab:ablation-semantic}
\centering
\vspace{-4pt}
\setlength{\tabcolsep}{4pt}
\begin{tabular}{lcc}
    \toprule
    \textbf{Method} & \texttt{Yahoo} & \texttt{IMDB} \\
    \midrule
    \textbf{Guided Diffusion (Full)} & \textbf{0.90}$_{\pm 0.04}$ & \textbf{0.59}$_{\pm 0.11}$ \\
    w/o semantic-aware re-masking & 0.79$_{\pm 0.02}$ & 0.35$_{\pm 0.10}$ \\
    w/o semantic-aware sampling & 0.80$_{\pm 0.02}$ & 0.34$_{\pm 0.10}$ \\
    \bottomrule
\end{tabular}
\vspace{-10pt}
\end{wraptable}

\textbf{\textit{Ablation Study.}}~To isolate the contribution of each guidance mechanism, we ablate semantic-aware sampling and semantic-aware re-masking on \texttt{Yahoo} and \texttt{IMDB}. When disabling one component, we replace it with its standard unguided counterpart (independent token sampling or low-confidence re-masking) while leaving the other component active, and report the mean cosine similarity over $1{,}000$ samples in \T~\ref{tab:ablation-semantic}.

The full framework achieves the best alignment on both datasets, and disabling either causes a sharp drop. The effect is most pronounced on \texttt{IMDB}, where the partial variants ($0.35$ and $0.34$) collapse toward the fully unguided baseline (see \F~\ref{fig:dis-ppl}). The two components yield comparable scores in isolation (e.g., $0.79$ vs.\ $0.80$ on \texttt{Yahoo}), indicating that they capture complementary signals: semantic-aware sampling steers the candidate selection toward the target manifold, while semantic-aware re-masking preserves semantically critical tokens across diffusion steps. Combining both restores the full alignment quality, confirming that the two stages are essential and reinforcing.

\section{Differentially Private Synthetic Data}
\label{sec:dpsyn}

In this section, we present our framework for generating differentially private (DP) synthetic text.
The DP synthetic data generation pipeline is illustrated in \F~\ref{fig:framework}.

The fundamental challenge in applying DP to natural language is that text is discrete and symbolic;
unlike numerical data, one cannot simply compute the ``average'' of multiple sentences and add DP noise to obfuscate individual contributions.
Consequently, prior works have been forced to rely on complex and unstable methods like DP-SGD~\citep{yu2022differentially} or use an LLM to synthesize a huge corpus of synthetic data and select according to private samples iteratively~\citep{pe2024}.

The utility of the latter approach, Private Evolution~\citep[PE;][]{pe2024}, is limited mainly by its generation pipeline rather than by the DP mechanism.
If DP noise were the main limitation, removing it would raise utility substantially.
Instead, on \texttt{PubMed}, PE improves only from $27.90\%$ at $\varepsilon=1$ to $28.60\%$ without DP noise ($\varepsilon=\infty$; \T~\ref{tab:sens-eps}).
PE's own results show the same pattern: moving from $\varepsilon=1$ to $\varepsilon=\infty$ changes most metrics only modestly, leaves some unchanged or lower, and does not close the gap to training on the original data (Table~2 in~\citet{pe2024}).
A weaker privacy guarantee therefore does not remove PE's proposal-and-selection bottleneck.
This pipeline is also costly: PE generates a large pool of candidate texts in every iteration and then selects among them, so with the same backbone a \emph{single PE iteration} already takes $1.53\times$ as long as Guided Diffusion (\S~\ref{subsec:dp-eval}).
These limitations motivate a generation paradigm that is both more efficient and achieves higher utility for DP synthetic text.

Our framework circumvents these limitations by shifting the privacy mechanism to the semantic embedding space.
While embeddings allow for straightforward DP aggregation (e.g., computing noisy centroids), converting these abstract vectors back into coherent text has historically been an ill-posed inverse problem.
Our key insight is that our Guided Diffusion algorithm (\A s~\ref{alg:sampling} and~\ref{alg:remasking}) effectively solves this ``inversion gap,'' acting as a robust decoder that reconstructs high-fidelity text conditioned on the DP embeddings.
This enables a streamlined ``embed-privatize-generate'' workflow that was previously infeasible.

\subsection{Framework for DP Synthetic Text Generation with Guided Diffusion}
As depicted in \F~\ref{fig:framework}, our framework contains three distinct stages.
We assume that the base DLM and the text embedding model contain no information about the private samples, i.e., they are not trained on the private samples, and that all pipeline hyperparameters (the number of clusters, the per-class sample budget, the generation hyperparameters, and the choice of encoder) are fixed before the private data are accessed. To ensure that the first assumption holds in our evaluation, we construct \texttt{ICLR26} from reviews published after these models were released (\S~\ref{subsec:dp-eval}); for \texttt{PubMed}, whose abstracts predate the DLMs, the formal guarantee is unaffected, but the empirical utility may partly reflect pretraining exposure.
We now introduce the three stages of DP synthetic text generation with Guided Diffusion:

\noindent\textbf{\textit{Semantic Embedding:}}~We project the sensitive private dataset $\mathcal{D} = \{x_1, \dots, x_N\}$ into a semantic space using a pre-trained text embedding model $\mathcal{E}$. This yields a set of continuous embedding vectors $\{E_1, \dots, E_N\}$.
These vectors are normalized with $L_2$ norm equal to $1$, simplifying the sensitivity bound of the aggregated embedding vectors.

\noindent\textbf{\textit{DP Aggregation:}}~We apply a DP aggregation mechanism to these vectors.
To avoid collapsing the diverse semantic information of the dataset into a single mean, we employ the practical DP clustering algorithm~\citep{chang2021differentially}.
This algorithm partitions the high-dimensional embeddings into semantically coherent groups and computes a set of representative noisy centroids $E_{dp}$ satisfying $(\varepsilon, \delta)$-DP.
For \emph{labeled} corpora (e.g., \texttt{ICLR26} in \S~\ref{subsec:dp-eval}, where each review carries a score), DP clustering is run \emph{class-conditionally}: separately within each label class on its embedding subset, yielding per-class noisy centroids $\{E_{dp,c}\}$. Each synthetic sample inherits the class label of its source per-class centroid, and the per-class sample budget is fixed in advance as a public hyperparameter. Since the per-class subsets are disjoint, the joint release of all centroids is $(\varepsilon, \delta)$-DP under add/remove adjacency by parallel composition (\S~\ref{sec:preliminaries}).
These centroids effectively capture the distribution of the private text while ensuring privacy protection for individual samples.
The only data-dependent objects that leave the private environment are the noisy centroids returned by the DP clustering algorithm (and, if the algorithm returns a data-dependent number of centroids, that number, which is part of its DP output). The requested number of clusters $K$ is a public constant, the number of synthetic samples generated per centroid is fixed uniformly in advance, and raw cluster sizes and cluster assignments are neither released nor used by the decoder.

\noindent\textbf{\textit{Guided Diffusion Generation:}}~Finally, we utilize the DP embedding $E_{dp}$ as the guidance signal ($E_{\text{tgt}}$) for our Guided Diffusion algorithm with DLM $\mathcal{M}_\theta$.
The model iteratively denoises a random sequence, steered by our semantic-aware sampling and re-masking strategies (\S~\ref{sec:guided-diffusion-design}), to generate an arbitrarily large set of synthetic texts $\mathcal{Y}$ that are semantically aligned with $E_{dp}$.
The original private samples are rigorously protected with $(\varepsilon, \delta)$-DP guarantees.

\begin{figure}[!t]
    \centering
\includegraphics[width=1.0\linewidth]{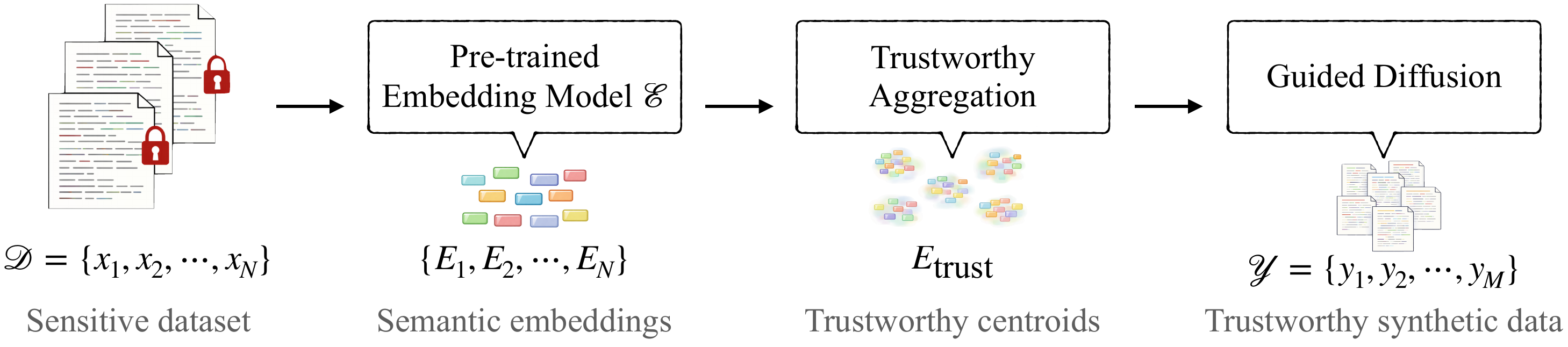}
    \caption{Trustworthy synthetic data generation with Guided Diffusion. For DP, the input is a private dataset, and the trustworthy aggregation is $(\varepsilon, \delta)$-DP aggregation. For DSI (\S~\ref{sec:dsisyn}), the input is a dataset collected from multiple, potentially malicious sources, and the trustworthy aggregation is DSI aggregation.}
    \label{fig:framework}
\end{figure}

\noindent\textbf{Privacy Guarantee.}~Under standard conditions of public pretrained $\mathcal{E}$ and $\mathcal{M}_\theta$, class-conditional DP aggregation via the DP clustering algorithm~\citep{chang2021differentially}, hyperparameters fixed before accessing $\mathcal{D}$, and a fixed generation procedure that does not access $\mathcal{D}$, the released synthetic dataset $\mathcal{Y}$ is $(\varepsilon, \delta)$-DP with respect to $\mathcal{D}$ under add/remove adjacency by post-processing invariance. The formal statement (Proposition~\ref{prop:dp-guarantee}) and proof are in Appendix~\ref{sec:appendix-proofs}.

\subsection{Evaluation: DP-Compliant Synthetic Text Generation}
\label{subsec:dp-eval}
To demonstrate the efficacy of our DP synthetic text paradigm with Guided Diffusion, we generate synthetic text under $(\varepsilon, \delta)$-DP guarantees and evaluate its downstream utility and its distributional fidelity to the private data.

\begin{table}[!t]
\caption{DP synthetic text generation at $(\varepsilon,\delta)=(1,10^{-6})$ with $2{,}000$ synthetic samples, reported as mean$_{\pm\text{std}}$ over four seeds. Utility is macro-F1 on \texttt{ICLR26} and causal next-token accuracy on \texttt{PubMed}. Bold marks the best result per dataset and backbone.}
\vspace{-5pt}
\label{tab:dp-main}
\centering
\small
\setlength{\tabcolsep}{3.5pt}
\begin{tabular}{llccccc}
    \toprule
    & & & \multicolumn{4}{c}{\textbf{Distributional fidelity}} \\
    \cmidrule(lr){4-7}
    \textbf{Backbone} & \textbf{Method} & \textbf{Utility} (\%) $\uparrow$ & FD $\downarrow$ & JSD $\downarrow$ & Dist-2 $\uparrow$ & MAUVE $\uparrow$ \\
    \midrule
    \multicolumn{7}{l}{\textit{\texttt{ICLR26} (utility: macro-F1)}} \\
    \multirow{3}{*}{\shortstack[l]{LLaDA-8B-\\Instruct}} & PE & 13.34$_{\pm 1.38}$ & 0.7753$_{\pm 0.0032}$ & 0.3401$_{\pm 0.0006}$ & 0.0730 & 0.6366$_{\pm 0.0054}$ \\
    & DP-SGD+LoRA & 20.12$_{\pm 4.16}$ & 0.6026$_{\pm 0.0015}$ & 0.3523$_{\pm 0.0007}$ & 0.0431 & 0.6475$_{\pm 0.0031}$ \\
    & \textbf{Guided Diffusion} & \textbf{24.85}$_{\pm 2.80}$ & \textbf{0.3995}$_{\pm 0.0020}$ & \textbf{0.2787}$_{\pm 0.0008}$ & \textbf{0.0891} & \textbf{0.6923}$_{\pm 0.0042}$ \\
    \cmidrule(lr){1-7}
    \multirow{3}{*}{Dream-7B} & PE & 12.49$_{\pm 0.57}$ & 0.7531$_{\pm 0.0041}$ & 0.3060$_{\pm 0.0005}$ & 0.0785 & 0.6409$_{\pm 0.0079}$ \\
    & DP-SGD+LoRA & \textbf{29.23}$_{\pm 3.37}$ & 0.6848$_{\pm 0.0020}$ & 0.3278$_{\pm 0.0006}$ & 0.0398 & 0.6281$_{\pm 0.0012}$ \\
    & \textbf{Guided Diffusion} & 28.68$_{\pm 2.40}$ & \textbf{0.4333}$_{\pm 0.0013}$ & \textbf{0.2620}$_{\pm 0.0006}$ & \textbf{0.0984} & \textbf{0.6848}$_{\pm 0.0045}$ \\
    \midrule
    \multicolumn{7}{l}{\textit{\texttt{PubMed} (utility: accuracy)}} \\
    \multirow{3}{*}{\shortstack[l]{LLaDA-8B-\\Instruct}} & PE & 22.15$_{\pm 0.08}$ & 0.5464$_{\pm 0.0039}$ & 0.3388$_{\pm 0.0005}$ & 0.1221 & 0.7318$_{\pm 0.0040}$ \\
    & DP-SGD+LoRA & 25.45$_{\pm 0.11}$ & 0.9021$_{\pm 0.0070}$ & 0.3432$_{\pm 0.0011}$ & 0.0786 & 0.6779$_{\pm 0.0021}$ \\
    & \textbf{Guided Diffusion} & \textbf{32.68}$_{\pm 0.08}$ & \textbf{0.3169}$_{\pm 0.0057}$ & \textbf{0.2061}$_{\pm 0.0007}$ & \textbf{0.2737} & \textbf{0.8637}$_{\pm 0.0017}$ \\
    \cmidrule(lr){1-7}
    \multirow{3}{*}{Dream-7B} & PE & 20.46$_{\pm 0.18}$ & 0.5801$_{\pm 0.0037}$ & 0.3144$_{\pm 0.0003}$ & 0.1282 & 0.7319$_{\pm 0.0071}$ \\
    & DP-SGD+LoRA & 25.16$_{\pm 0.14}$ & 0.9445$_{\pm 0.0052}$ & 0.3287$_{\pm 0.0008}$ & 0.0723 & 0.6804$_{\pm 0.0017}$ \\
    & \textbf{Guided Diffusion} & \textbf{32.89}$_{\pm 0.08}$ & \textbf{0.2387}$_{\pm 0.0039}$ & \textbf{0.1569}$_{\pm 0.0004}$ & \textbf{0.3014} & \textbf{0.8764}$_{\pm 0.0060}$ \\
    \bottomrule
\end{tabular}
\end{table}

\noindent\textbf{Experimental Setup.}~To isolate the generation algorithm from the capability of the underlying generator, we run every method on the same DLM backbone, either of two widely used DLMs, \texttt{LLaDA-8B-Instruct}~\citep{nie2025large} or \texttt{Dream-7B}~\citep{ye2025dream}, and compare against PE~\citep{pe2024} and \texttt{DP-SGD}+\texttt{LoRA}~\citep{li2024llmpbe}.
That is, PE and \texttt{DP-SGD}+\texttt{LoRA} also use the same DLM to generate synthetic text.
For each backbone, all of the methods use the same private data, privacy parameters, synthetic-data size, prompts, and downstream evaluation.

We evaluate on two private datasets. \texttt{PubMed}~\citep{yu2023training} contains $75{,}316$ medical abstracts; following the official PE protocol, we fine-tune \texttt{BERT-small}~\citep{turc2019well} on the synthetic DP samples and report causal next-token prediction accuracy on real held-out abstracts in \texttt{PubMed}'s test split.
To rule out the possibility that the training data of the evaluated models, including the DLMs and the embedding models, contains the private samples, we construct a new dataset, \texttt{ICLR26}, of paper reviews from ICLR 2026, which were released after all evaluated models.
We follow the setup and train/test split of the ICLR 2023 review dataset used by PE~\citep{pe2024}.
Specifically, we train a \texttt{roberta-base}~\citep{liu2019roberta} classifier on the synthetic reviews to predict the five-class review score and report macro-F1 on real held-out reviews in the test split.
Consistent with PE, we use the embedding model \texttt{sentence-t5-base} for \texttt{PubMed} and \texttt{stsb-roberta-base-v2} for \texttt{ICLR26}.
We run PE for four proposal-and-selection iterations, which already takes more than $4\times$ computational cost compared with Guided Diffusion.

Every method generates $2{,}000$ synthetic samples at $(\varepsilon,\delta)=(1,10^{-6})$, and results are averaged over four seeds.
Beyond downstream utility, we also measure distributional fidelity to the real data with four metrics: the Fr\'echet distance (FD) between MPNet~\citep{song2020mpnet} embeddings of synthetic and real texts, the Jensen--Shannon divergence (JSD) between their unigram distributions, Distinct-2~\citep{li2016diversity} for lexical diversity, and MAUVE~\citep{pillutla2021mauve}, which compares the two distributions in a language-model embedding space and penalizes both synthetic texts that are unlike real ones and real texts that the synthetic set fails to cover.

\noindent\textbf{Evaluation Results.}~In \T~\ref{tab:dp-main}, Guided Diffusion achieves the best score on every fidelity metric in all four dataset--backbone settings and the best utility in three of them; on \texttt{ICLR26} with \texttt{Dream-7B}, its macro-F1 is comparable with \texttt{DP-SGD}+\texttt{LoRA} (within standard deviation).

\textbf{\textit{Downstream Utility.}}~On \texttt{ICLR26} with \texttt{LLaDA-8B-Instruct}, Guided Diffusion reaches $24.85\%$ macro-F1, compared with $13.34\%$ for PE and $20.12\%$ for \texttt{DP-SGD}+\texttt{LoRA}. With \texttt{Dream-7B}, its macro-F1 ($28.68\%$) is $16.19$ points above PE and $0.55$ points below \texttt{DP-SGD}+\texttt{LoRA}. PE's macro-F1 stays near $13\%$ with both backbones.
On \texttt{PubMed}, Guided Diffusion improves causal next-token accuracy over PE by $10.53$ points with \texttt{LLaDA-8B-Instruct} ($32.68\%$ vs.\ $22.15\%$) and by $12.43$ points with \texttt{Dream-7B} ($32.89\%$ vs.\ $20.46\%$), and over \texttt{DP-SGD}+\texttt{LoRA} by more than $7$ points with both backbones.
The performance of Guided Diffusion also exceeds the results of PE's official implementation, where PE and \texttt{DP-SGD-FT}~\citep{yue2023synthetic} on \texttt{GPT-2 Large} use their default generators ($27.90\%$ and $31.10\%$; \T~\ref{tab:sens-eps}).
Our result approaches the non-private upper bound obtained by fine-tuning on the real private corpus ($34.60\%$).

\textbf{\textit{Distributional Fidelity.}}~Guided Diffusion produces synthetic text that is closer to the private data. It attains the lowest FD and unigram JSD and the highest Distinct-2 and MAUVE in every setting. The largest gap is in FD: with \texttt{LLaDA-8B-Instruct}, Guided Diffusion reduces FD from $0.78$ (PE) and $0.60$ (\texttt{DP-SGD}+\texttt{LoRA}) to $0.40$ on \texttt{ICLR26}, and from $0.55$ and $0.90$ to $0.32$ on \texttt{PubMed}. The higher Distinct-2 shows that this fidelity does not come from collapsing to a few prototypical texts.

\textbf{\textit{Computational Efficiency.}}~On eight A100 GPUs, Guided Diffusion generates the $2{,}000$ \texttt{ICLR26} samples in $2.74$ hours. End-to-end PE with the same backbone takes $4.19$ hours for \emph{a single} private-selection iteration and $12.20$ hours for the four iterations used in \T~\ref{tab:dp-main}, so Guided Diffusion is $4.45\times$ faster.
The difference is structural: PE repeatedly queries the generator for candidate variants and selects among them with a DP histogram, whereas our framework aligns each sample within a single denoising run.
In our implementation, DLM forward passes account for $47.6\%$ of the runtime, encoder scoring for $29.0\%$, and candidate construction for $17.0\%$; Appendix~\ref{sec:appendix-runtime} gives the full profile, including batch size, memory, and encoder-call counts.
This comparison holds for a shared DLM backbone; PE with an autoregressive generator that supports KV caching has a different cost profile, which we have not measured.
We do not include DP-SGD-based methods here, since their cost is dominated by private fine-tuning with per-sample gradient clipping and PE reports DP fine-tuning to be more than $10\times$ slower than PE~\citep{pe2024}.

\begin{wraptable}{r}{0.46\linewidth}
\vspace{-12pt}
\caption{Ablation of semantic-aware sampling and re-masking on \texttt{PubMed} with \texttt{LLaDA-8B-Instruct} under $\varepsilon=1$.}
\vspace{-5pt}
\label{tab:dp-ablation}
\centering
\small
\setlength{\tabcolsep}{3pt}
\begin{tabular}{lc}
    \toprule
    \textbf{Method} & \textbf{Accuracy (\%)} \\
    \midrule
    \textbf{Guided Diffusion (Full)} & \textbf{32.68} \\
    w/o semantic-aware sampling & 18.01 \\
    w/o semantic-aware re-masking & 26.73 \\
    \bottomrule
\end{tabular}
\vspace{-10pt}
\end{wraptable}

\textbf{\textit{Hyperparameter Sensitivity.}}~Guided Diffusion is insensitive to the privacy budget: on \texttt{PubMed} with \texttt{LLaDA-8B-Instruct}, its accuracy is $32.68\%$ / $32.65\%$ / $32.48\%$ at $\varepsilon = 1$ / $2$ / $4$, so the strict budget $\varepsilon=1$ costs little utility, and it stays above every DP baseline at each budget (\T~\ref{tab:sens-eps}). Generation quality is insensitive to the candidate pool size $\alpha$ and the weight ratio $b/a$ in \A~\ref{alg:sampling}.
See Appendix~\ref{sec:appendix-sensitivity} for more results.

\textbf{\textit{Ablation Study.}}~To verify that the gains observed in \S~\ref{sec:guided-diffusion-design} carry over to the downstream settings, we ablate each guidance mechanism on \texttt{PubMed} with \texttt{LLaDA-8B-Instruct} while keeping the same DP embedding pipeline. As shown in \T~\ref{tab:dp-ablation}, removing either component substantially degrades downstream accuracy: without semantic-aware sampling, accuracy collapses from $32.68\%$ to $18.01\%$, and without semantic-aware re-masking it drops to $26.73\%$. This confirms that both components remain essential in DP synthetic text generation.

\section{Trustworthy Synthetic Text via DSI}
\label{sec:dsisyn}

Beyond DP, we extend our Guided Diffusion framework to incorporate Data-Specific Indistinguishability~\citep[DSI;][]{xiao2025trustworthy}, providing complementary trustworthiness guarantees.
The overall pipeline is consistent with the DP synthetic data generation framework in \F~\ref{fig:framework}, with DP aggregation replaced by DSI-based aggregation.

DSI is designed to mitigate specific data risks with rigorous protection guarantees, including backdoor attacks, copyright infringement, and unintended memorization.
Unlike DP, which enforces indistinguishability against the presence of an arbitrary single sample under worst-case sensitivity, DSI ensures that the output is statistically indistinguishable from the output produced by a ``safe'' reference dataset.
In our context, this guarantees that synthetic data generated from protected embeddings is statistically close to data generated from references with potentially corrupted contributions removed.
Formally, DSI transfers high-probability trustworthy guarantees from a designated set of safe references to the target output:
\begin{definition}[DSI Gaussian Mechanism]
Given a function $\mathcal{F}$, a target input $U$, and reference datasets $R_{[1:m]} = \{R_1, \dots, R_m\}$, the mechanism $\mathcal{F}(U) + \mathbf{e}$ with $\mathbf{e} \sim \mathcal{N}(\bm{0}, \Sigma)$ satisfies $(\gamma_{[1:m]})$-DSI if, for a chosen divergence $\mathcal{D}_f$,
\begin{equation*}
\mathcal{D}_f\!\left(\mathbb{P}_{\mathcal{F}(U)+\mathbf{e}} \,\|\, \mathbb{P}_{\mathcal{F}(R_i)+\mathbf{e}}\right) \leq \gamma_i, \quad \forall i \in [1{:}m].
\end{equation*}
When $\mathcal{D}_f$ is the R\'enyi divergence of order $\alpha$, $\gamma_i$ recovers the $\varepsilon$ of $(\alpha, \varepsilon)$-R\'enyi DP per reference.
\end{definition}
DSI permits anisotropic $\Sigma$ optimized for the geometry of $(U, R_{[1:m]})$, which our framework exploits to construct high-utility trustworthy aggregations. Detailed motivation, the safe-reference construction, and the convex optimization for $\Sigma$ are deferred to Appendix~\ref{sec:appendix-dsi-background}.

\subsection{Framework for Trustworthy Synthetic Text Generation with Guided Diffusion}

To deploy DSI for trustworthy synthetic data generation, we consider a distributed data collection paradigm where the corpus is collected from multiple sources.
This setup mimics collaborative environments where data is aggregated from diverse providers, some of whom may be untrusted.

\noindent\textbf{\textit{Semantic Embedding:}}~Similar to the DP pipeline, we first project the collected text into a continuous semantic embedding space using a pre-trained semantic encoder, and then partition the embeddings into distinct clusters (e.g., via K-Means).
The number of clusters may be set to one, in which case aggregation reduces to a single average embedding over the collected samples.
After clustering, DSI-protected groups are defined locally within each cluster/class partition: for cluster $r$ and source $i$, $G_{r,i}=C_r\cap S_i$ denotes source $i$'s records inside cluster $r$; for labeled corpora, the same definition is applied within each class.
These groups are disjoint, so removing one group changes only the corresponding cluster-level DSI release; the joint synthetic release therefore preserves the same DSI bound by parallel composition and post-processing.

\noindent\textbf{\textit{DSI Aggregation:}}~Within each cluster, we define reference sets by removing groups.
For leave-$k$-groups-out DSI, each reference excludes any $k$ protected groups from the cluster.
We then compute the aggregated cluster centroid using the \textit{Optimal DSI Gaussian Mechanism}, which injects calibrated noise so that the released centroid is statistically indistinguishable from one computed from the corresponding safe reference, bounding the influence of corrupted groups on the aggregation.

\noindent\textbf{\textit{Guided Diffusion Generation:}}~Finally, these DSI-sanitized embeddings serve as the conditioning targets for Guided Diffusion.
The model generates synthetic text that preserves the semantic utility of the original cluster while inheriting the DSI indistinguishability guarantee through post-processing.

Our DSI Guided Diffusion framework is applicable to a broad spectrum of trustworthiness risks, including copyright protection and memorization mitigation as studied in the original DSI formulation~\citep{xiao2025trustworthy}.
Empirical evaluations against \textit{backdoor attacks} (data poisoning) and \textit{unintended memorization}~\citep{carlini2021extracting} are reported in Appendix~\ref{sec:appendix-dsi-backdoor} and Appendix~\ref{sec:appendix-memorization}, respectively.
In both cases, the pipeline without DSI already shows no observed attack at the tested scale ($0$ of $60$ triggered tests succeed; max ROUGE-L stays below the $0.2$ memorization threshold), and DSI Guided Diffusion preserves these outcomes at $\varepsilon=1$ and $\varepsilon=2$, respectively, with a small utility cost.
DSI therefore adds a formal certificate rather than an observed improvement: finite-sample tests cannot exclude untested triggers or rarer memorization events, whereas the DSI bound transfers the high-probability properties of each safe reference to the actual release by post-processing.

\noindent\textbf{Trustworthy Guarantee.}~Under analogous conditions of public pretrained $\mathcal{E}$ and $\mathcal{M}_\theta$, class/cluster-conditional DSI aggregation via the Optimal DSI Gaussian Mechanism, and a fixed generation procedure that does not access the collected corpus, the synthetic dataset $\mathcal{Y}$ remains $\gamma_{c,r,i}$-indistinguishable from the leave-one-protected-group-out reference $\mathcal{F}(R_{c,r,i})$ for every protected group $G_{c,r,i}$.
The guarantee transfers from the safe references to $\mathcal{Y}$ by DSI post-processing.
The formal statement (Proposition~\ref{prop:dsi-guarantee}) and proof sketch are in Appendix~\ref{sec:appendix-proofs}.

\section{Related Work}
\label{sec:relatedwork}
Prior DP synthetic text generation falls into two streams: \textit{DP-SGD-based fine-tuning}~\citep{yoon2018pategan, yu2021differentially, yue2023synthetic} and \textit{Private Evolution (PE) variants}~\citep{lin2024differentially, pe2024} that filter LLM candidates via DP histograms; both inherit structural limits from per-sample clipping or rejection sampling. The closest concurrent work, \citet{zhou2025private}, reconstructs DP-protected \emph{images} from privatized embeddings via diffusion-based image editing~\citep{rombach2022high}; translating this ``embedding-to-data'' approach to text is non-trivial because text is discrete and no off-the-shelf guided text-diffusion pipeline exists. Our framework bridges this gap via DLMs and accommodates both DP and DSI~\citep{xiao2025trustworthy}. A detailed elaboration is given in Appendix~\ref{sec:appendix-related-work}.

\section{Conclusion \& Limitations}
\label{sec:conclusion}

We propose a novel framework that converts semantic embeddings into high-fidelity text via Guided Diffusion, establishing a new paradigm for private synthetic text generation.
Our approach naturally integrates Differential Privacy, ensuring precise semantic control while providing rigorous privacy guarantees with little loss in utility.
The same framework extends to Data-Specific Indistinguishability, bounding the influence of corrupted sources and individual samples at a small utility cost.
Our empirical results demonstrate the superiority of Guided Diffusion over prior DP synthetic text generation baselines across datasets and diffusion backbones.

\noindent\textbf{Limitations and Future Work.}~Several extensions remain open. Semantic-aware re-masking evaluates one ablated embedding per token; sparse or batched ablations could amortize this cost for very long sequences. Our generation pipeline operates on fixed-length token blocks; extending to longer documents fits naturally with hierarchical or sliding-window schemes that we leave for future work. Empirically, we focus on English short-form generation; transferring to multilingual corpora, code, and structured outputs is a promising direction.

\subsection*{AI Use Statement}
We used generative AI tools to assist with the writing of this paper, mainly to check grammar and to polish the wording and readability of author-written text.
The tools were also used to format \LaTeX{} tables and figures.
All research ideas, the Guided Diffusion method and its technical components, the privacy analysis, the experimental design and implementation, and the interpretation of results were developed by the authors.
We did not use generative AI to generate datasets or to develop the method or the conceptual framework.
The authors reviewed and revised all AI-polished text.
The authors take responsibility for all the final content of this work.

\subsection*{Ethics Statement}
This work aims to reduce the privacy risks in synthetic text generation.
The DP and DSI guarantees hold under the assumptions stated in Propositions~\ref{prop:dp-guarantee} and~\ref{prop:dsi-guarantee}; practitioners should verify these assumptions before deployment.
DP bounds the influence of each individual record but does not prevent inferences about population-level properties of the private corpus.
Our experiments use publicly available data: \texttt{PubMed} abstracts, \texttt{Yahoo Answers Topics}, \texttt{IMDB} and \texttt{Yelp} reviews, and ICLR 2023 and ICLR 2026 paper reviews publicly posted on OpenReview.
We use the reviews only to evaluate synthetic data generation and make no attempt to identify reviewers.
Our study involves no human-subject experiments.
We discuss broader impacts in Appendix~\ref{sec:appendix-broader-impact}.

\subsection*{Reproducibility Statement}
We describe Guided Diffusion in \S~\ref{sec:guided-diffusion-design}, with pseudocode in \A s~\ref{alg:sampling} and~\ref{alg:remasking}, the DP synthetic text framework in \S~\ref{sec:dpsyn}, and the DSI instantiation in \S~\ref{sec:dsisyn}, with its experiments in Appendices~\ref{sec:appendix-dsi-backdoor} and~\ref{sec:appendix-memorization}.
Appendix~\ref{sec:implementation} lists all hyperparameters, the backbone and embedding models, and the DP clustering configuration, and Appendix~\ref{sec:appendix-sensitivity} reports the sensitivity to the privacy budget, the number of synthetic samples, and the generation hyperparameters.
All backbone DLMs, embedding models, and baseline implementations are publicly available, and \S~\ref{subsec:dp-eval} describes the evaluation datasets and protocols; \texttt{ICLR26} follows the setup and train/test split of the ICLR 2023 review dataset used by PE~\citep{pe2024}.
We will release our code publicly.

\bibliography{main}
\bibliographystyle{iclr2027_conference}

\newpage
\appendix
\section{Broader Impact}
\label{sec:appendix-broader-impact}
This paper presents work whose goal is to advance the field of privacy-preserving machine learning. Our framework for private synthetic text generation has significant potential to benefit society by enabling the safe utilization of sensitive datasets, such as medical records, financial reports, and private communications, that are currently inaccessible due to privacy regulations (e.g., GDPR). By reconciling high-fidelity generation with rigorous DP guarantees, our approach mitigates critical risks associated with Large Language Models, including data memorization and leakage from training corpora.

While advancements in generative AI carry inherent risks of misuse, such as the generation of misinformation or impersonating content, our research focuses specifically on defensive mechanisms that certify safety. The DP guarantee bounds the influence of any single record on the released synthetic data, providing a quantifiable foundation for responsible release of synthetic corpora derived from sensitive sources. We hope this work contributes to the responsible and transparent deployment of AI systems in sensitive domains.

\section{Detailed Related Work}
\label{sec:appendix-related-work}

\noindent\textbf{Guided Diffusion Models.}~Image diffusion models have achieved remarkable success in controllable synthesis, using embedding-based guidance to steer generation~\citep{song2019generative, ho2020denoising, song2021denoising, dhariwal2021diffusion, rombach2022high}.
While extensive research optimizes this guidance for images~\citep{ramesh2022hierarchical, ye2023ip, zhang2023adding, feng2023trainingfree, bansal2024universal}, its application to DLMs remains underexplored.
Unlike the continuous image domain, text is discrete, and the dominance of prompt-based AR LLMs has overshadowed the development of embedding-based control.
To the best of our knowledge, we are the first to effectively guide DLMs in converting semantic embeddings into coherent, high-fidelity text.

\noindent\textbf{Differentially Private Synthetic Data Generation.}~Prior research on generating synthetic data with DP guarantees spans both image and text domains~\citep{hu2024sok}, generally bifurcating into two primary streams: 1) \textit{DP-SGD-based methods}~\citep{yoon2018pategan, yu2021differentially, li2022large, harder2023pretrained, yue2023synthetic, ghalebikesabi2023differentially, dockhorn2023differentially, tan2025synthesizing} and 2) \textit{Private Evolution (PE) variants}~\citep{lin2024differentially, pe2024, pmlr-v235-hou24c, zhang2025pcevolve, lin2025differentially}.

DP-SGD approaches involve training/fine-tuning a generative model. While foundational, these methods suffer from significant computational overhead due to per-sample gradient clipping and often yield suboptimal utility when scaling to LLMs. In contrast, PE variants circumvent direct training by iteratively querying public LLMs to generate a massive pool of synthetic candidates. A DP selection mechanism, which adds Gaussian noise to distance histograms, is then employed to filter for samples that align with the private data distribution. While PE offers improved efficiency and utility-privacy tradeoffs compared to DP-SGD, it remains heavily dependent on the quality of manual sampling prompts and implicitly constrains the diversity and volume of synthetic samples due to the rejection-sampling bottleneck.

Recently, \citet{zhou2025private} introduced a novel paradigm for \textit{image} synthesis, demonstrating that high-fidelity images can be reconstructed from DP embeddings using diffusion-based image editing~\citep{rombach2022high}. However, translating this ``embedding-to-data'' approach to the text domain is non-trivial due to the discrete nature of language and the absence of established Guided Diffusion pipelines for text. Our work bridges this gap: we propose an alternative approach that, for the first time, exploits the steerable properties of DLMs to efficiently convert semantic embeddings into high-utility synthetic text. Crucially, this framework naturally accommodates DP mechanisms, providing rigorous privacy guarantees alongside superior generation quality.

\section{Background on DSI}
\label{sec:appendix-dsi-background}

\noindent\textbf{Motivation for DSI.}~DP is rooted in \emph{Input-Independent Indistinguishability} (III): the indistinguishability of the released leakage must hold for arbitrary pairs of neighboring datasets. III is necessary when the goal is to uniformly upper bound the additional knowledge an adversary gains from observing the leakage, but is overly conservative for trustworthiness concerns whose threat models are not worst-case over all neighbors.

\noindent\textbf{DSI as Transferring Guarantees from Safe References.}~DSI~\citep{xiao2025trustworthy} relaxes III by establishing \emph{specific} indistinguishability between the target output and a designated set of \emph{safe reference} outputs. As a concrete example, consider an input dataset $\mathcal{X}=\{\vecx_1, \vecx_2, \dots, \vecx_n\}$ where each $\vecx_i$ is contributed by source $i$, and suppose $\vecx_{i_0}$ is poisoning data from an unknown malicious source. The leave-one-out subset $\mathcal{X} \setminus \vecx_{i_0}$ is then a safe reference: a model $\mathcal{M}(\mathcal{X} \setminus \vecx_{i_0})$ trained on it should be robust to the backdoor in $\vecx_{i_0}$ with high probability. If we ensure indistinguishability between $\mathcal{M}(\mathcal{X})$ and \emph{each} $\mathcal{M}(\mathcal{X} \setminus \vecx_i)$ for $i=1,\dots,n$ -- which naturally includes the unknown $i_0$ -- then the trustworthy guarantee of $\mathcal{M}(\mathcal{X} \setminus \vecx_{i_0})$ transfers to $\mathcal{M}(\mathcal{X})$. Crucially, only the $n$ specific pairs $(\mathcal{M}(\mathcal{X}), \mathcal{M}(\mathcal{X} \setminus \vecx_i))$ need to be indistinguishable, rather than arbitrary pairs.

\noindent\textbf{Optimal DSI Gaussian Noise.}~Unlike the DP Gaussian mechanism, which adds isotropic noise calibrated to the global sensitivity, DSI permits anisotropic noise optimized for the specific geometry of $(U, R_{[1:m]})$. The optimal covariance $\Sigma$ minimizing the expected noise magnitude $\mathbb{E}[\|\mathbf{e}\|_2^2]$ subject to the DSI constraints is the solution of a convex optimization over $m$ Lagrangian multipliers; we use this in our framework (\S~\ref{sec:dsisyn}).

\section{Implementation Details}
\label{sec:implementation}

We list the concrete hyperparameter values used throughout our experiments. All settings are held fixed unless stated otherwise.
All experiments are run on A100 GPUs.

\subsection{Diffusion Backbone}
We use \texttt{LLaDA-8B-Instruct}~\citep{nie2025large} and \texttt{Dream-7B}~\citep{ye2025dream} as the backbone masked diffusion language models. For these models, every generation uses the settings below.
\begin{itemize}[itemsep=1pt, leftmargin=24pt, topsep=3pt]
    \item Sequence length: $L = 448$ tokens. Note that LLaDA takes a semi-autoregressive approach, with default block size $32$.
    \item We use the default number of diffusion (denoising) steps: $T = L = 448$.
    \item Denoise schedule $\rho_t$: We take the default linear schedule used by \texttt{LLaDA-8B-Instruct} and \texttt{Dream-7B}.
\end{itemize}

\subsection{Semantic-Aware Sampling (\A~\ref{alg:sampling})}
\begin{itemize}[itemsep=1pt, leftmargin=24pt, topsep=3pt]
    \item Candidate pool factor: $\alpha = 20$.
    \item Top-$N$ filter size: $N = 32$.
    \item Log-probability weight: $a = 1$.
    \item Hamming-distance weight: $b = 0.1$.
\end{itemize}
We did not tune $a$ and $b$ extensively; Appendix~\ref{sec:appendix-sensitivity} shows that varying $b/a \in \{0.1, 0.2, 0.5\}$ and $\alpha \in \{5, 10, 20\}$ has little effect on generation quality.

\subsection{Semantic-Aware Re-Masking (\A~\ref{alg:remasking})}
\begin{itemize}[itemsep=1pt, leftmargin=24pt, topsep=3pt]
    \item Denoise schedule $\rho_t$: We take the same default linear schedule used by the backbone \texttt{LLaDA-8B-Instruct} and \texttt{Dream-7B}.
\end{itemize}

\subsection{DP and DSI Aggregation}
\begin{itemize}[itemsep=1pt, leftmargin=24pt, topsep=3pt]
    \item DP clustering: \citet{chang2021differentially}'s implementation with default hyperparameters; the privacy budget $(\varepsilon, \delta)$ is specified per experiment. The number of clusters is set to $100$. For each class, the number of synthetic samples generated from each released centroid is fixed uniformly in advance. We do not use private cluster sizes.
    \item DSI optimal Gaussian noise: computed via the convex optimization formulation described in Appendix~\ref{sec:appendix-dsi-background} (see also~\citet{xiao2025trustworthy}), with the security parameter $\varepsilon$ specified per experiment ($\varepsilon=1$ in Appendix~\ref{sec:appendix-dsi-backdoor}, $\varepsilon=2$ in Appendix~\ref{sec:appendix-memorization}).
\end{itemize}

\subsection{Embedding Models}
We deliberately use two different semantic embedding models chosen to match the dominant semantic axis of each dataset. \texttt{sentence-t5-base}~\citep{ni2022sentence}, trained on a broad mixture of sentence-level tasks, produces embeddings that emphasize topical and factual content; we therefore use it for \texttt{PubMed} (medical abstracts grouped by topic area) and \texttt{Yahoo} (10 explicit topic categories). \texttt{stsb-roberta-base-v2}~\citep{stsbrobert}, fine-tuned on the STS-Benchmark sentence-similarity task that emphasizes evaluative and stylistic similarity, captures sentiment- and style-driven semantics more reliably, so we use it for the review-heavy datasets \texttt{ICLR26}, \texttt{ICLR23}, \texttt{Yelp}, and \texttt{IMDB}, where rating and stylistic register dominate the relevant signal.
\begin{itemize}[itemsep=1pt, leftmargin=24pt, topsep=3pt]
    \item \texttt{ICLR26}, \texttt{ICLR23}, \texttt{Yelp}, \texttt{IMDB}: \texttt{stsb-roberta-base-v2}~\citep{stsbrobert}.
    \item \texttt{PubMed}, \texttt{Yahoo}: \texttt{sentence-t5-base}~\citep{ni2022sentence}.
\end{itemize}

\section{Sensitivity Analyses}
\label{sec:appendix-sensitivity}

\textbf{\textit{Privacy Budget.}}~\T~\ref{tab:sens-eps} varies $\varepsilon \in \{1, 2, 4\}$ with $\delta=10^{-6}$ using \texttt{LLaDA-8B-Instruct} on \texttt{PubMed}.
The accuracy of Guided Diffusion changes by at most $0.2$ points, so the strict budget $\varepsilon=1$ used in our main results costs little utility relative to looser budgets. At $\varepsilon=1$, Guided Diffusion is within $2$ points of the non-private upper bound of fine-tuning on the real private corpus ($34.60\%$). For comparison, the table also lists the baselines across budgets: PE and \texttt{DP-SGD-FT}~\citep{yue2023synthetic} on \texttt{GPT-2 Large} from PE's official implementation~\citep{pe2024}, which use their default generators, and \texttt{DP-SGD}+\texttt{LoRA} on \texttt{Llama-3-8B}~\citep{li2024llmpbe}. Guided Diffusion exceeds every DP baseline at each tested budget, and even PE without DP noise ($\varepsilon=\infty$, $28.60\%$). Scaling DP-SGD to \texttt{Llama-3-8B} with LoRA does not help: its accuracy stays within $24.4\%$--$24.9\%$ across budgets, below \texttt{DP-SGD-FT} on \texttt{GPT-2 Large} at $\varepsilon=1$.

\begin{table}[!t]
\caption{Effect of the privacy budget on \texttt{PubMed} accuracy, with $\delta=10^{-6}$. Guided Diffusion uses \texttt{LLaDA-8B-Instruct}. $^\dagger$Results from PE's official implementation, where each baseline uses its default generator.}
\label{tab:sens-eps}
\centering
\small
\begin{tabular}{lcc}
    \toprule
    \textbf{Method} & $\boldsymbol{\varepsilon}$ & \textbf{Accuracy (\%)} \\
    \midrule
    \multicolumn{2}{l}{Without fine-tuning} & 0.07 \\
    \multicolumn{2}{l}{Fine-tuning on real data (non-DP)} & 34.60 \\
    \midrule
    \multirow{2}{*}{PE$^\dagger$} & $1$ & 27.90 \\
    & $\infty$ & 28.60 \\
    \midrule
    \multirow{2}{*}{DP-SGD-FT (GPT-2 Large)$^\dagger$} & $1$ & 31.10 \\
    & $\infty$ & 33.10 \\
    \midrule
    \multirow{3}{*}{DP-SGD+LoRA (Llama-3-8B)} & $1$ & 24.88 \\
    & $2$ & 24.70 \\
    & $4$ & 24.37 \\
    \midrule
    \multirow{3}{*}{\textbf{Guided Diffusion}} & $1$ & \textbf{32.68} \\
    & $2$ & \textbf{32.65} \\
    & $4$ & \textbf{32.48} \\
    \bottomrule
\end{tabular}
\end{table}

\textbf{\textit{Number of Synthetic Samples.}}~Since the synthetic texts are a post-processing of the released DP centroids, generating more of them incurs no additional privacy cost. \T~\ref{tab:sens-size} shows that utility improves as the synthetic training set grows from $100$ to $2{,}000$ samples: macro-F1 rises from $13.13\%$ to $24.85\%$ on \texttt{ICLR26} with \texttt{LLaDA-8B-Instruct}.

\textbf{\textit{Generation Hyperparameters.}}~\T~\ref{tab:sens-hparam} varies the candidate pool factor $\alpha$, the weight ratio $b/a$ of the sampling score, and the number of denoising steps $T$. We report the cosine similarity between each generated text and its target embedding under the guidance encoder (Guide sim.) and under an independent MPNet encoder (MPNet sim.), together with perplexity (PPL). Pool size and weighting have little effect: across $\alpha \in \{5, 10, 20\}$ and $b/a \in \{0.1, 0.2, 0.5\}$, guide similarity stays within $0.691$--$0.693$ and perplexity within $4.20$--$4.48$. Halving the denoising steps to $T=224$ lowers guide similarity from $0.692$ to $0.676$ and nearly doubles perplexity, from $4.22$ to $7.95$.
This degradation with fewer denoising steps is expected: it is a property of the DLM rather than of our guidance design.
We therefore recommend using the default number of denoising steps of each DLM.

\begin{table}[!t]
\caption{Effect of the number of synthetic samples on \texttt{ICLR26} (\texttt{LLaDA-8B-Instruct}, $\varepsilon=1$), reported as mean$_{\pm\text{std}}$.}
\label{tab:sens-size}
\centering
\small
\begin{tabular}{rc}
    \toprule
    \textbf{Synthetic samples} & \textbf{Macro-F1 (\%)} \\
    \midrule
    $100$ & 13.13$_{\pm 5.59}$ \\
    $500$ & 20.46$_{\pm 6.30}$ \\
    $1{,}000$ & 20.83$_{\pm 5.71}$ \\
    $2{,}000$ & \textbf{24.85}$_{\pm 2.80}$ \\
    \bottomrule
\end{tabular}
\end{table}

\begin{table}[!t]
\caption{Effect of generation hyperparameters on \texttt{ICLR26}. Each row changes one hyperparameter from the default.}
\label{tab:sens-hparam}
\centering
\small
\begin{tabular}{lccc}
    \toprule
    \textbf{Setting} & \textbf{Guide sim.} $\uparrow$ & \textbf{MPNet sim.} $\uparrow$ & \textbf{PPL} $\downarrow$ \\
    \midrule
    Default ($T=448$, $\alpha=20$, $b/a=0.1$) & 0.6920 & 0.5396 & 4.218 \\
    $T=224$ & 0.6756 & 0.5335 & 7.948 \\
    $\alpha=5$ & 0.6929 & 0.5345 & 4.482 \\
    $\alpha=10$ & 0.6930 & 0.5405 & 4.352 \\
    $b/a=0.2$ & 0.6919 & 0.5385 & 4.247 \\
    $b/a=0.5$ & 0.6910 & 0.5415 & 4.203 \\
    \bottomrule
\end{tabular}
\end{table}

\section{Runtime Profile}
\label{sec:appendix-runtime}
\T~\ref{tab:runtime} profiles Guided Diffusion with \texttt{LLaDA-8B-Instruct} at batch size $16$ in BF16 on a single A100-SXM4-40GB, the same configuration used for PE.
One output uses one complete denoising trajectory of $T=448$ DLM forward passes; the encoder is called in batched form on the filtered candidates and on the token-ablated variants, which amounts to $897$ logical encoder calls (executed as $10{,}527$ internal micro-batches) per batch.
Generating $2{,}000$ samples takes $21.94$ hours on one A100 and $2.74$ hours on eight, versus $12.20$ hours for four rounds of PE with the same backbone on the same eight GPUs.

\begin{table}[!t]
\caption{Runtime profile of Guided Diffusion (\texttt{LLaDA-8B-Instruct}, batch size $16$, BF16, one A100-SXM4-40GB).}
\label{tab:runtime}
\centering
\small
\begin{tabular}{lr}
    \toprule
    \textbf{Measurement} & \textbf{Value} \\
    \midrule
    Wall-clock time per $16$-sample batch & 631.86 s \\
    Amortized time per sample & 39.49 s \\
    $2{,}000$ samples on 1 / 8 A100s & 21.94 h / 2.74 h \\
    Peak allocated memory & 16.35 GiB \\
    DLM forward passes per batch & 448 \\
    Logical encoder calls / internal micro-batches per batch & 897 / 10{,}527 \\
    \midrule
    \textbf{Component} & \textbf{Share of time} \\
    \midrule
    DLM forward passes & 47.6\% \\
    Encoder forward passes & 29.0\% \\
    Candidate construction & 17.0\% \\
    Decoding / other & 3.7\% / 2.7\% \\
    \bottomrule
\end{tabular}
\end{table}

\section{Evaluation: Backdoor-Resistant Synthetic Text Generation}
\label{sec:appendix-dsi-backdoor}
We evaluate DSI Guided Diffusion against a semantic backdoor attack in synthetic text generation.

\noindent\textbf{Experimental Setup.}~
For \textit{dataset}, we use \texttt{ICLR23}, the dataset of ICLR 2023 paper reviews that the authors of PE~\citep{pe2024} collected from OpenReview, to aggregate semantic embeddings and generate synthetic reviews.
The downstream utility is evaluated by training a classifier on $2{,}000$ synthetic samples for review-rating prediction; accuracy is measured on the \texttt{ICLR23} test split.
Unlike the DP evaluation in \S~\ref{subsec:dp-eval}, which treats each dataset as a single private corpus, here the data is partitioned across $m=3$ sources, one of which may be malicious.
For labeled data, aggregation is performed class-conditionally, and synthetic samples inherit the rating label of their source class/cluster centroid.
For \textit{models}, we employ \texttt{LLaDA-8B-Instruct} as the backbone DLM and \texttt{stsb-roberta-base-v2} as the semantic embedding encoder; we fine-tune \texttt{roberta-base}~\citep{liu2019roberta} as the text classifier.

\textit{Threat Model.}~
We consider a distributed setting where data is contributed from $m=3$ sources.
Within each class/semantic cluster, the DSI protected groups are the source-specific subsets, matching the protected-group definition in Proposition~\ref{prop:dsi-guarantee}.
One source may be malicious and inject poisoned reviews to manipulate the review-rating classifier, while the other sources provide non-poisoned data.
The attacker uses a semantic trigger: it generates paper reviews on the topic of medical images and labels them as ``strong reject.''
If the attack succeeds, reviews about medical images are classified as strong reject at inference time.
The attack success rate (ASR) is measured on a held-out poisoned test set; the attacker injects $45$ poisoned samples into the training corpus, and ASR is evaluated on $60$ triggered test samples.

\textit{Defense Mechanism and Baselines.}~
We apply DSI Guided Diffusion with leave-one-group-out aggregation.
For each class/cluster centroid, we compute the optimal DSI Gaussian noise required to make the released centroid indistinguishable from one computed without any single source group in that class/cluster.
We set the DSI security parameter to $\varepsilon=1$ (see \S~\ref{sec:dsisyn}).
We compare four methods: (1) \texttt{Orig-benign}, training the classifier on the original benign dataset; (2) \texttt{Orig-backdoor}, training on the original dataset with the backdoor planted; (3) \texttt{Synthetic-w/o-DSI}, training on synthetic data generated by Guided Diffusion without DSI noise; and (4) \texttt{Synthetic-w/-DSI}, training on synthetic data generated using DSI Guided Diffusion.

\noindent\textbf{Experimental Results.}~
The results are summarized in \T~\ref{tab:dsi-backdoor}.
We evaluate clean accuracy on benign data (ACC) and backdoor attack success rate (ASR).

\textit{Vulnerability of Standard Training.}~
As shown in the first two rows of \T~\ref{tab:dsi-backdoor}, training on the original benign data yields an ACC of $47.80\%$ and 0\% ASR.
When the backdoor is planted, the model remains useful ($48.82\%$ ACC) but becomes fully compromised, exhibiting 100\% ASR.
This confirms that standard training is highly susceptible to this semantic poisoning attack.

\textit{Empirical Robustness of Guided Diffusion.}~
Guided Diffusion already provides strong empirical mitigation even without explicit DSI noise.
Generating synthetic data from raw aggregated embeddings drops ASR to 0\% on the $60$ triggered test examples.
This suggests that semantic embedding aggregation followed by Guided Diffusion reconstruction disrupts this topic-triggered backdoor pattern.

\textit{Certified Mitigation with DSI.}~
While Guided Diffusion itself mitigates this attack empirically, DSI adds a formal indistinguishability certificate relative to the leave-one-group-out reference.
With $\varepsilon=1$, DSI Guided Diffusion also maintains 0\% ASR while achieving an ACC of $37.74\%$, a small utility drop relative to \texttt{Synthetic-w/o-DSI} ($39.85\%$).
For context, this accuracy is comparable to PE without DP noise ($38.5\%$ using \texttt{Llama-2-7B}), while additionally providing the DSI certificate.

\begin{table}[!t]
    \centering
    \caption{Evaluation against a semantic backdoor attack. ASR is measured on $60$ held-out triggered examples. DSI Guided Diffusion preserves the empirical 0\% ASR while adding a formal indistinguishability certificate.}
    \label{tab:dsi-backdoor}
    \small
    \setlength{\tabcolsep}{4pt}
    \begin{tabular}{lcc}
        \toprule
        \textbf{Method} & \textbf{ACC (\%)} & \textbf{ASR (\%)} \\
        \midrule
        \texttt{Orig-benign} & 47.80 & 0 \\
        \texttt{Orig-backdoor} & 48.82 & 100 \\
        \midrule
        \texttt{Synthetic-w/o-DSI} & \textbf{39.85} & \textbf{0} \\
        \textbf{\texttt{Synthetic-w/-DSI}} & \textbf{37.74} & \textbf{0} \\
        \bottomrule
    \end{tabular}
    \vspace{-4pt}
\end{table}

\section{Evaluation: Memorization-Mitigated Synthetic Text Generation}
\label{sec:appendix-memorization}
Beyond data poisoning, DSI is also designed to mitigate \textit{unintended memorization}, where a generative model regenerates verbatim training samples~\citep{carlini2021extracting}. We evaluate whether our DSI Guided Diffusion can certifiably bound such memorization on a sensitive review corpus.

\noindent\textbf{Experimental Setup.}~We adopt the same backbone and embedding models as in Appendix~\ref{sec:appendix-dsi-backdoor} and modify the threat model to target memorization.
For \textit{dataset}, we use $50$ 1-star Grocery reviews randomly drawn from \texttt{Yelp}~\citep{yelp2023}.

\textit{Threat Model.}~An adversary attempts to extract training samples via prefix-based completion: for each training sample, we provide the first $100$ tokens as a prefix and generate $10$ completions, then measure the maximum ROUGE-L between any generated suffix and the true suffix. Values below $0.2$ indicate low memorization.

ROUGE-L~\citep{lin2004rouge} is an F-score based on the longest common subsequence between two token sequences and is widely used to quantify near-verbatim overlap; in our setting, a higher ROUGE-L between a generated suffix and the true suffix indicates more verbatim or near-verbatim copying of the training sample.
For each prefix we report the maximum ROUGE-L over the $10$ generated completions, using $0.2$ as a conservative cutoff below which generations are considered to share no significant verbatim content with the reference.

\textit{Defense Mechanism and Baselines.}~We apply DSI Guided Diffusion with leave-one-out aggregation at the individual sample level (there is a single cluster in this setting), calibrating the optimal DSI Gaussian noise so that the generated centroid is statistically indistinguishable from one computed without any single training sample.
We set the security parameter to $\varepsilon=2$.
We compare against \texttt{Synthetic-w/o-DSI}, which conditions Guided Diffusion on the average sample embedding without any DSI noise, isolating the empirical contribution of guided generation.

\noindent\textbf{Experimental Results.}~\T~\ref{tab:dsi-memorization} reports the max ROUGE-L value averaged over the samples (lower is better) and semantic alignment (cosine similarity to the target embedding; higher is better). The semantic alignment values reported here are higher than those in \S~\ref{sec:guided-diffusion-design} (e.g., $0.59$ on \texttt{IMDB}) because the 1-star Grocery subset of \texttt{Yelp} is far more topically homogeneous than \texttt{IMDB}, which spans diverse movies, genres, and stylistic registers.

\textit{Empirical Memorization Mitigation by Guided Diffusion.}~Both variants achieve max ROUGE-L well below the $0.2$ memorization threshold ($0.1446$ for \texttt{Synthetic-w/o-DSI} and $0.1460$ for \texttt{Synthetic-w/-DSI}), confirming that Guided Diffusion empirically suppresses verbatim copying of training samples even before any DSI noise is added. This mirrors the Appendix~\ref{sec:appendix-dsi-backdoor} finding that the framework offers strong empirical robustness on its own.

\textit{Certified Memorization Mitigation by DSI.}~The cost of converting empirical robustness into a formal indistinguishability guarantee is small: applying DSI Gaussian noise at $\varepsilon=2$ keeps max ROUGE-L small while keeping semantic alignment high ($0.76$ vs.\ $0.80$ without DSI).
Crucially, DSI extends the empirical guarantee to a formal indistinguishability bound: it certifies closeness to the leave-one-out centroid-generation pipeline while preserving low observed copying. It bounds the influence of individual samples under the specified DSI mechanism and reference construction.

\begin{table}[!t]
    \centering
    \caption{Memorization evaluation on \texttt{Yelp} 1-star Grocery reviews. Max ROUGE-L below $0.2$ indicates low memorization. DSI Guided Diffusion provides formal indistinguishability with a small utility cost. Results are reported as mean$_{\pm\text{std}}$.}
    \label{tab:dsi-memorization}
    \small
    \setlength{\tabcolsep}{4pt}
    \begin{tabular}{lcc}
        \toprule
        \textbf{Method} & \textbf{Max ROUGE-L $\downarrow$} & \textbf{Semantic Alignment $\uparrow$} \\
        \midrule
        \texttt{Synthetic-w/o-DSI} & 0.1446$_{\pm 0.0223}$ & 0.80$_{\pm 0.07}$ \\
        \textbf{\texttt{Synthetic-w/-DSI}} & 0.1460$_{\pm 0.0187}$ & 0.76$_{\pm 0.07}$ \\
        \bottomrule
    \end{tabular}
    \vspace{-4pt}
\end{table}

\section{Formal Guarantees and Proof Sketches}
\label{sec:appendix-proofs}

\begin{proposition}[DP Guarantee]
\label{prop:dp-guarantee}
Let $\mathcal{D} = \{(x_i, y_i)\}_{i=1}^N$ denote the private corpus, where $y_i$ is an optional class label, and let $\mathcal{Y}$ denote the labeled synthetic dataset produced by the DP Guided Diffusion framework. Suppose the following conditions hold:
\begin{enumerate}[itemsep=1pt, leftmargin=14pt, topsep=2pt]
    \item \textbf{Public pretrained models.} The embedding model $\mathcal{E}$ and the diffusion language model $\mathcal{M}_\theta$ are pretrained on a corpus disjoint from $\mathcal{D}$ and are treated as public.
    \item \textbf{Data-independent configuration.} The requested number of clusters $K$, the per-class synthetic sample budget, the generation hyperparameters, and the choice of $\mathcal{E}$ and $\mathcal{M}_\theta$ are fixed before $\mathcal{D}$ is accessed.
    \item \textbf{Class-conditional DP aggregation.} The DP clustering algorithm~\citep{chang2021differentially} is run separately within each label $c$, releasing noisy centroids $E_{dp,c}$ (and, if the algorithm returns a data-dependent number of centroids, that number) that are $(\varepsilon, \delta)$-DP with respect to add/remove adjacency on the per-class subset $\mathcal{D}_c = \{x_i : y_i = c\}$. (For an unlabeled corpus, this reduces to a single class.) Nothing else derived from $\mathcal{D}$, in particular no raw cluster size or assignment, is released.
    \item \textbf{Fixed generation procedure.} Guided Diffusion (Algorithms~\ref{alg:sampling} and~\ref{alg:remasking}) is a randomized function of $\{E_{dp,c}\}$ and the public models alone, and does not access $\mathcal{D}$. Each synthetic sample inherits the class label $c$ of its source per-class centroid.
\end{enumerate}
Then the released synthetic dataset $\mathcal{Y} = \{(\tilde{x}_j, c_j)\}$ is $(\varepsilon, \delta)$-DP with respect to $\mathcal{D}$ under add/remove adjacency. Under replace-one adjacency, where a single record may change its label, $\mathcal{Y}$ is $(2\varepsilon, 2\delta)$-DP.
\end{proposition}

\begin{proof}[Proof sketch of Proposition~\ref{prop:dp-guarantee} (DP Guarantee)]
By condition~(3), each per-class centroid release $E_{dp,c}$ is $(\varepsilon, \delta)$-DP w.r.t.\ $\mathcal{D}_c$. Under add/remove adjacency a single record belongs to exactly one class, so the per-class subsets $\{\mathcal{D}_c\}$ form a disjoint partition of $\mathcal{D}$ and parallel composition yields $(\varepsilon, \delta)$-DP for the joint release $\{E_{dp,c}\}$ w.r.t.\ $\mathcal{D}$; under replace-one adjacency a record may leave one class and enter another, affecting two independent mechanisms, and basic composition gives $(2\varepsilon, 2\delta)$. Conditions~(1), (2), and~(4) ensure that $\mathcal{Y}$ is a randomized function of $\{E_{dp,c}\}$ and public model components only, with no further dependence on $\mathcal{D}$. Post-processing invariance (\S~\ref{sec:preliminaries}) preserves the $(\varepsilon, \delta)$-DP guarantee on the final synthetic data release.
\end{proof}

\begin{proposition}[DSI Guarantee]
\label{prop:dsi-guarantee}
Let $\mathcal{X}=\{(x_i,y_i)\}$ denote the input corpus, where $y_i$ is an optional class label.
After optional class-conditional clustering, let $C_{c,r}$ denote cluster $r$ within class $c$; for an unlabeled corpus, there is a single class.
Let $S_s$ denote the records from source $s$.
We define DSI protected groups locally within each class/cluster partition:
\[
    G_{c,r,s}=C_{c,r}\cap S_s .
\]
The nonempty $G_{c,r,s}$ form a disjoint partition of the clustered corpus.
For leave-$k$-groups-out DSI, let $A\subseteq [m_{c,r}]$ with $|A|=k$, define
\[
    G_{c,r,A}=\bigcup_{s\in A}G_{c,r,s},
    \qquad
    R_{c,r,A}=C_{c,r}\setminus G_{c,r,A}.
\]
Let $\mathcal{Y}$ denote the synthetic dataset produced by DSI Guided Diffusion. Suppose the following conditions hold:
\begin{enumerate}[itemsep=1pt, leftmargin=14pt, topsep=2pt]
    \item \textbf{Public pretrained models.} The embedding model $\mathcal{E}$ and the diffusion language model $\mathcal{M}_\theta$ are fixed public mechanisms independent of $\mathcal{X}$, and all hyperparameters, including the clustering configuration, are fixed before $\mathcal{X}$ is accessed.
    \item \textbf{Class/cluster-conditional DSI aggregation.} For each class/cluster pair $(c,r)$, the DSI Gaussian mechanism releases a sanitized centroid
    \[
        E_{\mathrm{dsi},c,r}=F(C_{c,r})+e_{c,r}
    \]
    satisfying
    \[
        D_f\!\left(P_{F(C_{c,r})+e_{c,r}}\,\middle\|\,P_{F(R_{c,r,A})+e_{c,r}}\right)
        \le \gamma_{c,r,A}
    \]
    for every leave-$k$-groups-out reference $R_{c,r,A}$. The randomness used for different class/cluster releases is independent. When $D_f$ is Rényi divergence, we report the corresponding DSI security parameter $\varepsilon=\max_{c,r,A}\gamma_{c,r,A}$.
    \item \textbf{Fixed generation procedure.} The Guided Diffusion procedure $\mathcal{G}$ is a randomized function of $\{E_{\mathrm{dsi},c,r}\}$ and the public models alone, and does not access $\mathcal{X}$. Each synthetic sample inherits the label of its source class/cluster centroid.
\end{enumerate}
Then, for every protected group collection $G_{c,r,A}$, the released synthetic dataset $\mathcal{Y}$ satisfies the $\gamma_{c,r,A}$-DSI indistinguishability bound with respect to the leave-$k$-groups-out reference $R_{c,r,A}$.
In particular, if any protected cluster-source group, or any $k$ protected groups within the same class/cluster partition, are corrupted, $\mathcal{Y}$ remains statistically close to the dataset that would have been produced from the corresponding safe reference.
\end{proposition}

\begin{proof}[Proof sketch of Proposition~\ref{prop:dsi-guarantee}]
By condition~(2), each class/cluster centroid release $E_{\mathrm{dsi},c,r}$ satisfies the stated DSI bound with respect to its leave-$k$-groups-out references.
Consider any protected group collection $G_{c,r,A}$ and the corresponding reference $R_{c,r,A}$.
Removing $G_{c,r,A}$ changes only the input to the $(c,r)$ cluster-level DSI mechanism; all other class/cluster releases have identical distributions under the target and reference inputs.
Since the class/cluster releases use independent randomness, the divergence of the joint centroid release equals the divergence of the changed coordinate:
\[
    D_f\!\left(
    P_{\{E_{\mathrm{dsi},c',r'}\}_{c',r'}}
    \,\middle\|\,
    P_{\{E_{\mathrm{dsi},c',r'}^{(-G_{c,r,A})}\}_{c',r'}}
    \right)
    =
    D_f\!\left(
    P_{F(C_{c,r})+e_{c,r}}
    \,\middle\|\,
    P_{F(R_{c,r,A})+e_{c,r}}
    \right)
    \le \gamma_{c,r,A}.
\]
Thus the joint release of all DSI-sanitized centroids preserves the same bound by parallel composition over disjoint class/cluster partitions.
Conditions (1) and~(3) ensure that the final synthetic dataset $\mathcal{Y}$ is a randomized function only of the DSI-sanitized centroids and public model components, with no further dependence on $\mathcal{X}$.
The post-processing property of DSI~\citep{xiao2025trustworthy} therefore preserves the indistinguishability bound on the final synthetic text release.
\end{proof}

\end{document}